\documentclass[11pt]{article}
\usepackage{fullpage}
\usepackage{authblk}
\usepackage{graphicx}
\usepackage{amsmath}
\usepackage{amsfonts}
\usepackage{amssymb}
\usepackage{amsthm}
\usepackage{mathrsfs}
\usepackage{bm}
\usepackage{mathtools}
\usepackage{enumitem}
\usepackage{algorithm2e}

\usepackage{url}
\usepackage{hyperref}

\usepackage{tcolorbox}
\newtcolorbox{wbox}
{
	colback  = white,
}
\usepackage{caption}
\usepackage{subcaption}
\newtheorem{theorem}{Theorem}%[section]
\newtheorem{lemma}[theorem]{Lemma}

\newtheorem{corollary}[theorem]{Corollary}

\newtheorem{example}{Example}

\newcommand{\continuation}{??}

\newtheorem*{claim*}{Claim}
\newtheorem*{remark*}{Remark}

\newtheorem*{definition*}{Definition}

\makeatletter
\renewcommand\section{%
  \@startsection{section}{1}
                {\z@}%
                {-3.5ex \@plus -1ex \@minus -.2ex}%
                {2.3ex \@plus.2ex}%
                {\large\bfseries}% 11pt
}
\renewcommand\subsection{%
  \@startsection{subsection}{2}
                {\z@}%
                {-3.25ex\@plus -1ex \@minus -.2ex}%
                {1sp}% No space after subsections
                {\normalsize\bfseries}% normal size, boldface
}
\renewcommand\subsubsection{%
  \@startsection{subsubsection}{3}
                {\z@}%
                {-3.25ex\@plus -1ex \@minus -.2ex}%
                {1sp}% No space after subsubsections
                {\normalfont\normalsize}% normal size, medium
}
\makeatother

\title{{\Large\bf On the Population Monotonicity of Independent Set Games}}

\author[1]{Libing Wang}
\author[1]{Han Xiao
\thanks{Corresponding author. Email: hxiao@ouc.edu.cn.}}
\author[2]{Donglei Du}
\author[3]{Dachuan Xu}

\affil[1]{School of Mathematical Sciences, Ocean University of China, Qingdao, China}
\affil[2]{Faculty of Management, University of New Brunswick, Fredericton, New Brunswick, E3B 5A3, Canada}
\affil[3]{Department of Operations Research and Information Engineering, Beijing University of Technology, Beijing, China}

\date{}
\begin{document}

%\numberwithin{equation}{section}

\maketitle

\openup 1.2\jot

%\hfill

\begin{abstract}
An independent set game is a cooperative game dealing with profit sharing in the maximum independent set problem.
A population monotonic allocation scheme is a rule specifying how to share the profit of each coalition among its participants such that every participant is better off when the coalition expands.
In this paper, we provide a necessary and sufficient characterization for independent set games admitting population monotonic allocation schemes.
Moreover, our characterization can be verified efficiently.

\hfill

\noindent\textbf{Keywords:} Cooperative game $\cdot$ Population monotonic allocation scheme $\cdot$ Independent set

\noindent\textbf{Mathematics Subject Classification:}  05C57 $\cdot$ 91A12 $\cdot$ 91A43 $\cdot$ 91A46
% 05C57 Combinatorics - Games on graphs
% 91A12 Game Theory - Cooperative games
% 91A43 Game Theory - Games involving graphs
% 91A46 Game Theory - Combinatorial games
\hfill

\hfill
\end{abstract}

%\newpage

\section{Introduction}
An independent set game is a cooperative game with graph structure.
The players are edges and the value of each coalition is defined to be the maximum size of independent sets induced by the coalition.
The following scenario captures the essence of independent set games.
There are projects and participants.
Every participant is suitable for two projects but only allowed to join one.
Every project requires all suitable participants to cooperate to be done.
The aim is to finish as many projects as possible.
We may introduce a graph to represent projects, participants and their relations, where every vertex is a project and every edge is a participant joining two suitable projects.
Then the aim becomes finding a maximum independent set.
And the independent set game deals with profit sharing in this scenario.

The core is a central concept in profit sharing, which contains all possible ways of distributing the total worth of a game among individual participants such that no participant derives a better payoff by leaving either individually or as a group.
Thus the core emphasizes the stability of the game.
However, the core does not necessarily guarantee the unhindered formation of a coalition, as a new member may decrease the profit shared by participants in the current coalition.
To retain stability in an expanding coalition, population monotonic allocation schemes (PMASes for short) were introduced \cite{Spru90}, under which no participant derives a worse payoff when a new member joins the coalition.
Thus, PMASes can be viewed as a refinement of the core, where the latter emphasizes the static stability and the former emphasizes the dynamic stability.

Combinatorial optimization games take an important part in cooperative game theory and draw a lot of attention from researchers.
Faigle and Kern \cite{FK93} studied the approximate core.
Deng et al. \cite{DIN99} provided a unified characterization for the core non-emptiness.
Universal characterizations for the core non-emptiness of each subgame were studied in \cite{DINZ00, ELN12, DLN13}.
Immorlica et al. \cite{IMM08} studied approximate PMASes.
Research efforts have also been made for individual combinational optimization games, including matching games \cite{SS71, KP03, KPT20, Vazi22, XF22}, flow games \cite{KZ82, DFS06}, spanning tree games \cite{NMT04, KS20}, etc.
For independent set games, core non-emptiness conditions for the game itself and every subgame were successfully settled in \cite{DIN99, ELN12, DLN13}.
In addition, Xiao et al. \cite{XWF21} provided a necessary and sufficient characterization for convexity and Liu et al. \cite{LXF21} gave a combinatorial characterization for PMASes in convex instances.
By investigating the interplay of PMASes and underlying structures, we provide a necessary and sufficient characterization for PMASes in independent set games.
Our characterization identifies a larger class of independent set games than the work of Xiao et al. \cite{XWF21},
since every convex game is also population monotonic \cite{Spru90} but the converse is not necessarily true.
Moreover, our characterization can be verified efficiently.

The remainder of this paper is organized as follows.
In Section \ref{sec:preliminaries}, some notions and notations used in this paper are introduced.
Section \ref{sec:PM_ISGame} develops a necessary and sufficient characterization for PMASes of independent set games.
Section \ref{sec:discussion} compares the characterizations for convexity and population monotonicity of independent set games and discusses the direction of further research.

\section{Preliminaries}
\label{sec:preliminaries}
\subsection{Game theory}
A \emph{cooperative game} $\Gamma=(N,\gamma)$ consists of a \emph{player set} $N$ and a \emph{characteristic function} $\gamma:2^N\rightarrow \mathbb{R}$ with convention $\gamma(\emptyset)=0$.
We call $N$ the \emph{grand coalition} and call $S$ a \emph{coalition} for any $S\subseteq N$.
The game $\Gamma$ is \emph{convex} if $\gamma (S)+\gamma (T)\leq \gamma (S\cap T) + \gamma (S\cup T)$ for any $S,T\subseteq N$.

An \emph{allocation} of the game $\Gamma$ is a non-negative vector $\boldsymbol{x}=(x_i)_{i\in N}$ specifying how to distribute the profit among players in the grand coalition $N$.
The \emph{core} of the game $\Gamma$ is the set of allocations $\boldsymbol{x}=(x_i)_{i\in N}$ satisfying \emph{efficiency} and \emph{coalitional rationality} conditions,
\begin{enumerate}
  \item[\textendash] \emph{efficiency}: $\sum_{i\in N} x_{i}=\gamma(N)$;
  \item[\textendash] \emph{coalitional rationality}: $\sum_{i\in S}x_i
\geq \gamma(S)$ for any $S\subseteq N$.
\end{enumerate}

An \emph{allocation scheme} of the game $\Gamma$ is a collection of non-negative vectors $\{\boldsymbol{x}_S\}_{S\in 2^N \backslash \{\emptyset\}}$ with $\boldsymbol{x}_S=(x_{S,i})_{i\in S}$ specifying how to distribute the profit among players in every coalition $S\in 2^N\backslash \{\emptyset\}$.
A \emph{population monotonic allocation scheme} (PMAS) is an allocation scheme $\{\boldsymbol{x}_S\}_{S\in 2^N \backslash \{\emptyset\}}$ satisfying \emph{efficiency} and \emph{monotonicity} conditions,
\begin{enumerate}
  \item[\textendash] \emph{efficiency}: $\sum_{i\in S} x_{S,i}=\gamma(S)$ for any $S\in 2^N\backslash \{\emptyset\}$;
  \item[\textendash] \emph{monotonicity}: $x_{S,i}\leq x_{T,i}$ for any $S, T\in 2^N\backslash \{\emptyset\}$ with $S\subseteq T$ and any $i\in S$.
\end{enumerate}
The game $\Gamma$ is \emph{population monotonic} if it admits a PMAS.
Sprumont \cite{Spru90} proved that every convex game is population monotonic by providing an incremental procedure for constructing a PMAS from any convex game. 

\subsection{Graph theory}
All graphs considered in this paper are finite, undirected and simple.
Let $G=(V,E)$ be a graph.
For any vertex $v$ in $G$, we use $\delta(v)$ to denote the set of edges incident to $v$ and use $N(v)$ to denote the set of vertices adjacent to $v$.
For any two vertices $u$ and $v$ in $G$, we use $N(u,v)$ to denote the set of vertices that are only adjacent to $u$ and $v$, i.e., $N(u,v)=\big\{w\in V: N(w)=\{u,v\}\big\}$.
A vertex is \emph{isolated} if it has degree zero.
A vertex is \emph{pendant} if it has degree one.
An edge is \emph{pendant} if it has a pendant endpoint.
A vertex \emph{covers} all the edges incident to it.
For $U\subseteq V$, $G[U]$ denotes the induced subgraph of $G$ on $U$.
For $F\subseteq E$,
$V\langle F\rangle $ denotes the set of vertices incident \emph{only} to edges in $F$.
An \emph{independent set} of $G$ is a vertex set $U\subseteq V$ such that $G[U]$ has no edge.
We use $\alpha(G)$ to denote the size of maximum independent sets in $G$.

\section{Population monotonicity of independent set games}
\label{sec:PM_ISGame}
The \emph{independent set game} on a graph $G=(V,E)$ is a cooperative game $\Gamma_G=(N,\gamma)$ such that $N=E$ and $\gamma(S)=\alpha(G[V\langle S\rangle])$ for any $S\subseteq N$.
Throughout this paper, we always assume that the underlying graph $G$ has no isolated vertex.
In the following, we develop a complete characterization for the population monotonicity of $\Gamma_G$ via the structure of $G$.

\subsection{A glimpse of PMASes}
Notice that any independent set naturally yields an edge set decomposition.
Consider the underlying graph $G=(V,E)$ of the independent set game $\Gamma_G=(N,\gamma)$.
Let $I\subseteq V$ be an independent set and $\bar{E}(I)\subseteq E$ be the set of edges not covered by $I$.
Then $\{\delta (v):v\in I\}$ yields a partition for $E\backslash \bar{E}(I)$ since $\delta(u)\cap \delta(v)=\emptyset$ for any two vertices $u,v\in I$.
The edge set decomposition induced by independent sets offers a basic observation for PMASes in independent set games.

\begin{lemma}
\label{thm:decomposition}
Let $\{\boldsymbol{x}_S\}_{S\in 2^N\backslash \{\emptyset\}}$ be a PMAS of $\Gamma_G$.
Let $I$ be a maximum independent set of $G$.
Then $\sum_{i\in \delta (v)} x_{N,i} =\gamma\big(\delta (v)\big)=1$ for any $v\in I$.
\end{lemma}
\begin{proof}
Let $\bar{E}(I)$ be the set of edges not covered by $I$.
By definition of PMASes, we have
\begin{equation}
  \label{eq:decomposition}
  \begin{split}
  \gamma(N)
  &= \sum_{i\in \bar{E}(I)}x_{N,i}+ \sum_{v\in I}\sum_{i\in \delta(v)}x_{N,i}\\
  &\geq \sum_{v\in I} \sum_{i\in \delta (v)}x_{\delta (v),i}\\
  &=\sum_{v\in I}\gamma \big(\delta(v)\big)\\
  &=\lvert I\rvert.
  \end{split}
\end{equation}
The first equality in \eqref{eq:decomposition} is from the edge set decomposition induced by $I$.
The last equality in \eqref{eq:decomposition} is because $\gamma \big(\delta (v)\big)= 1$ for any $v\in I$.
It follows that
\begin{equation*}\label{eq:incident.edges}
\sum_{i\in \delta (v)} x_{N,i} =\gamma\big(\delta (v)\big)=1
\end{equation*}
for any $v\in I$.
\end{proof}

Lemma \ref{thm:decomposition} implies the following corollary directly.

\begin{lemma}
\label{thm:edge-exposed}
Let $\{\boldsymbol{x}_S\}_{S\in 2^N\backslash \{\emptyset\}}$ be a PMAS of $\Gamma_G$.
Then $0\leq x_{S,i}\leq 1$ for any $S\subseteq N$ with $i\in S$.
Moreover, if edge $i$ is not covered by some maximum independent set of $G$, then $x_{S,i}=0$ for any $S\subseteq N$ with $i \in S$.
\end{lemma}

\begin{proof}
Let $I$ be a maximum independent set.
Let $v$ be a vertex in $I$.
Lemma \ref{thm:decomposition} implies that $0\leq x_{N,i}\leq 1$ for any $i\in \delta (v)$.
Let $\bar{E}(I)$ be the set of edges not covered by $I$.
Lemma \ref{thm:decomposition} also implies that $x_{N,i}=0$ for any $i\in \bar{E}(I)$ as
\begin{equation*}
  \begin{split}
  \sum_{i\in \bar{E}(I)} x_{N,i}
  &= \gamma (N)-\sum_{v\in I} \sum_{i\in \delta (v)} x_{N,i}\\
  &=\gamma (N) - \lvert I\rvert\\
  &=0.
  \end{split}
\end{equation*}
By monotonicity of PMASes, $0\leq x_{S,i}\leq x_{N,i}\leq 1$ for any $S\subseteq N$ with $i\in S$.
In particular, if $i\in \bar{E}(I)$, then $x_{S,i}=x_{N,i}=0$ for any $S\subseteq N$ with $i\in S$.
Consequently, if edge $i$ is not covered by some maximum independent set of $G$, then $x_{S,i}=0$ for any $S\subseteq N$ with $i\in S$.
\end{proof}

Lemma \ref{thm:edge-exposed} suggests that an edge is a null player in the independent set game if it is not covered by some maximum independent set.
Indeed, the game has the same value with or without the edge that is not covered by some maximum independent set.
Consequently, any core allocation defined from a PMAS yields a null payoff to a null player.
In contrast, every pendant edge is covered by all maximum independent sets, which suggests that pendant edges are essential in independent set games.
Next, we study pendant edges in PMASes of independent set games.

\subsection{Pendant edges in PMASes}

We have the following observation for pendant edges in PMASes.

\begin{lemma}
\label{thm:edge-pendant}
Let $\{\boldsymbol{x}_S\}_{S\in 2^N\backslash \{\emptyset\}}$ be a PMAS of $\Gamma_G$.
If edge $i$ is pendant in $G$, then $x_{S,i}=1$ for any $S\subseteq N$ with $i\in S$.
\end{lemma}

\begin{proof}
Let $v$ be a pendant endpoint of edge $i$ in $G$.
There exists a maximum independent set $I$ of $G$ with $v\in I$.
To see this, consider any maximum independent set $I$.
If $v\not\in I$, the other endpoint $u$ of edge $i$ is in $I$.
It follows that $I-u+v$ is a desired maximum independent set of $G$.
By efficiency of PMASes, we have $\gamma \big(\delta (v)\big)=x_{i,i}=1$.
Lemma \ref{thm:edge-exposed} and monotonicity of PMASes imply that $x_{S,i}=1$ for any $S\subseteq N$ with $i\in S$.
\end{proof}

Lemma \ref{thm:edge-pendant} shows that pendant edges are important in the independent set game.
Indeed, every maximum independent set covers all pendant edges and the game value decreases when pendant edges are removed. 
It turns out that, for population monotonic independent set games, not only pendant edges are necessary, but also every non-pendant edge has to be close to pendant edges.

\begin{lemma}
\label{thm:necessity-1}
If $\Gamma_G$ is population monotonic, then every vertex has distance at most two to a pendant vertex in $G$.
\end{lemma}
\begin{proof}
Assume to the contrary that vertex $v^*$ has distance larger than two to pendant vertices in $G$.

We first show that there is at most one neighbor $u\in N (v^*)$ with $\lvert N (u, v^*)\rvert \leq 1$.
Let $u\in N (v^*)$ be a vertex with $\lvert N(u,v^*)\rvert\leq 1$.
Let $\{\boldsymbol{x}_S\}_{S\in 2^N\backslash \{\emptyset\}}$ be a PMAS of $\Gamma_G$.
By definition of PMASes, we have
\begin{equation}
\label{eq:key-inequality}
\gamma \big(\delta(u)\cup \delta(v^*)\big) + x_{\delta(v^*),u v^*}\geq \gamma \big(\delta(u)\big) + \gamma \big(\delta(v^*)\big).
\end{equation}
Since $\delta (u)\cup \delta (v^*)$ contains no pendant edge, $\gamma\big(\delta(u)\big)=\gamma \big(\delta(v^*)\big)=1$.
Moreover, $\lvert N (u, v^*)\rvert\leq 1$ implies $\gamma \big(\delta(u)\cup \delta(v^*)\big)=1$ as $V\langle \delta(u)\cup \delta(v^*) \rangle$ induces either an edge or a triangle.
Then $x_{\delta(v^*),uv^*} \geq 1$ follows from \eqref{eq:key-inequality}.
Since $\gamma\big(\delta (v^*)\big)=1$, efficiency of PMASes implies that there is at most one vertex $u\in N (v^*)$ with $\lvert N (u, v^*)\rvert \leq 1$.

Notice that $\lvert N(v^*)\rvert\geq 2$, as vertex $v^*$ is not pendant.
Since at most one neighbor $u\in N(v^*)$ satisfies $\lvert N (u, v^*)\rvert \leq 1$,
there is a neighbor $u^* \in N (v^*)$ with $\lvert N (u^*,v^*)\rvert \geq 2$.
Notice that $\lvert N(w, v^*)\rvert \leq 1$ for any $w\in N (u^*,v^*)$, implying that there is more than one neighbor $u\in N (v^*)$ with $\lvert N (u, v^*)\rvert \leq 1$.
A contradiction occurs.
Hence, every vertex has distance at most two to a pendant vertex in $G$.
\end{proof}

\begin{figure}[h]
\centering
\hspace{0em}\includegraphics[width=.4\textwidth]{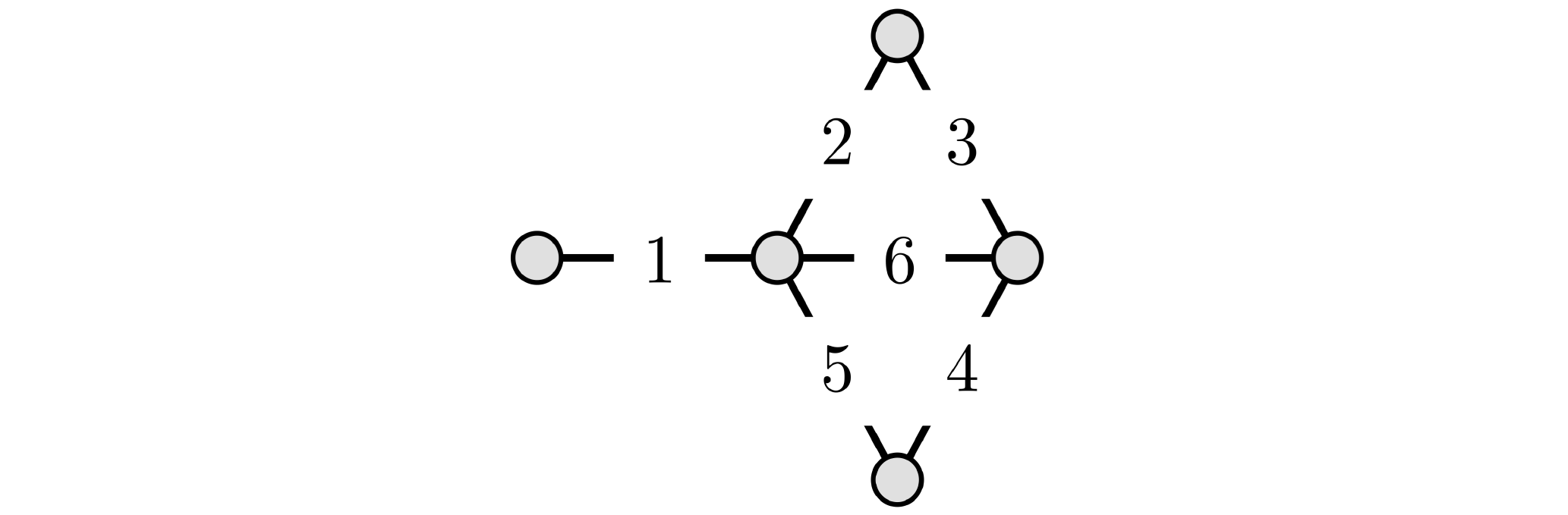}
\caption{A minimal counterexample for the converse of Lemma \ref{thm:necessity-1}.}
\label{fig:CounterExample1}
\end{figure}

\begin{example}
\label{ex:1}
Let $\Gamma=(N,\gamma)$ be the independent set game defined on the graph in Figure \ref{fig:CounterExample1}.
Clearly, every vertex in the underlying graph has distance at most two to a pendant vertex.
However, $\Gamma$ is not population monotonic.
Assume to the contrary that $\{\boldsymbol{x}_S\}_{S\in 2^{N} \backslash \{\emptyset\}}$ is a PMAS of $\Gamma$.
Notice that
\begin{equation*}
\gamma(23)=\gamma(2346)=\gamma(45)=\gamma(3456)=1.
\end{equation*}
By definition of PMASes, we have
\begin{equation*}
0\leq x_{2346,4}+x_{2346,6}\leq \gamma (2346)-\gamma (23)=0,
\end{equation*}
and
\begin{equation*}
0\leq x_{3456,3}+x_{3456,6}\leq \gamma (3456)-\gamma (45)=0.
\end{equation*}
By monotonicity of PMASes, we have $x_{346,3}=x_{3456,3}=0$, $x_{346,4}=x_{2346,4}=0$ and $x_{346,6}=x_{2346,6}=x_{3456,6}=0$.
It follows that
\begin{equation*}
0=x_{346,3}+x_{346,4}+x_{346,6}<\gamma (346)=1,
\end{equation*}
which contradicts efficiency of PMASes.
$\lhd$
\end{example}

Example \ref{ex:1} shows that Lemma \ref{thm:necessity-1} is not sufficient for the population monotonicity of independent set games.
It also suggests that the obstruction might lie in non-pendant edges.
Next, we turn to study non-pendant edges in PMASes of independent set games.

\subsection{Non-pendant edges in PMASes}
Before proceeding, we decompose non-pendant edges into two categories according to whether they are incident to pendant edges.
A non-pendant edge is said of \emph{type I} if it is not incident to any pendant edge, and of \emph{type II} if it is incident to a pendant edge.
Clearly, every non-pendant edge is either of type I or of type II, but not both.
We have the following observations for non-pendant edges in PMASes.

\begin{lemma}
\label{thm:NonPendantEdgeI}
Let $\{\boldsymbol{x}_S\}_{S\in 2^N\backslash \{\emptyset\}}$ be a PMAS of $\Gamma_G$. 
If $i$ is a non-pendant edge of type I with endpoints $u$ and $v$, then $x_{S,i}=1$ for any $S\subseteq N$ with $\delta (u)\subseteq S$ or $\delta (v)\subseteq S$.
\end{lemma}
\begin{proof}
By Lemma \ref{thm:necessity-1}, every vertex has distance at most two to pendant vertices  in $G$ .
Since $i$ is a non-pendant edge of type I, both endpoints $u$ and $v$ have distance two to pendant vertices.
It follows that $N(u,v)=\emptyset$ as every vertex in $N(u,v)$ has distance larger than two to pendant vertices.
Hence $\gamma \big(\delta (u)\cup \delta (v)\big)=1$.
Further notice that $ \gamma \big(\delta (u)\big) = \gamma \big(\delta (v)\big) =1$.
By definition of PMASes, we have
\begin{equation*}
\gamma \big(\delta (u)\cup \delta (v)\big) + x_{\delta (u),i}\geq \gamma \big(\delta (u)\big) + \gamma \big(\delta (v)\big).
\end{equation*}
It follows that $x_{\delta (u),i}=1$.
Similarly, we have $x_{\delta (v),i}=1$.
By monotonicity of PMASes, $x_{S,i}=1$ for any $S\subseteq N$ with $\delta (u)\subseteq S$ or $\delta (v)\subseteq S$.
\end{proof}

\begin{lemma}
\label{thm:NonPendantEdgeII}
Let $\{\boldsymbol{x}_S\}_{S\in 2^N\backslash \{\emptyset\}}$ be a PMAS of $\Gamma_G$.
If $i$ is a non-pendant edge of type II that is incident to a non-pendant edge of type I,
then $x_{S,i}=0$ for any $S\subseteq N$ with $i\in S$.
\end{lemma}
\begin{proof}
Without loss of generality, we may assume that edge $i$ is covered by every maximum independent set of $G$, since otherwise the assertion follows from Lemma \ref{thm:edge-exposed} directly.
Let $u$ and $v$ be the endpoints of edge $i$.
Further assume that $u$ has distance one and $v$ has distance two to pendant vertices.
Let $I$ be a maximum independent set of $G$ that contains as many pendant vertices as possible.
It follows that $u\not\in I$ as $u$ is adjacent to a pendant vertex in $G$.
Since edge $i$ is covered by $I$, $u\not \in I$ implies $v\in I$.
Let $i^*$ be a non-pendant edge of type I that is incident to $i$.
Clearly, $i^*\in \delta (v)$.
Lemma \ref{thm:NonPendantEdgeI} implies $x_{N,i^*}=1$.
Then $x_{N,i}=0$ follows from Lemma \ref{thm:decomposition}.
By monotonicity of PMASes, $x_{S,i}=0$ for any $S\subseteq N$ with $i\in S$.
\end{proof}

With Lemmas \ref{thm:NonPendantEdgeI} and \ref{thm:NonPendantEdgeII}, we are able to deduce the value in PMASes for special non-pendant edges.
Another minimal counterexample for the converse of Lemma \ref{thm:necessity-1} was detected with those observations for non-pendant edges.
Moreover, the counterexample suggests another necessary condition for the population monotonicity of independent set games.

\begin{figure}[h]
\centering
\includegraphics[width=.4\textwidth]{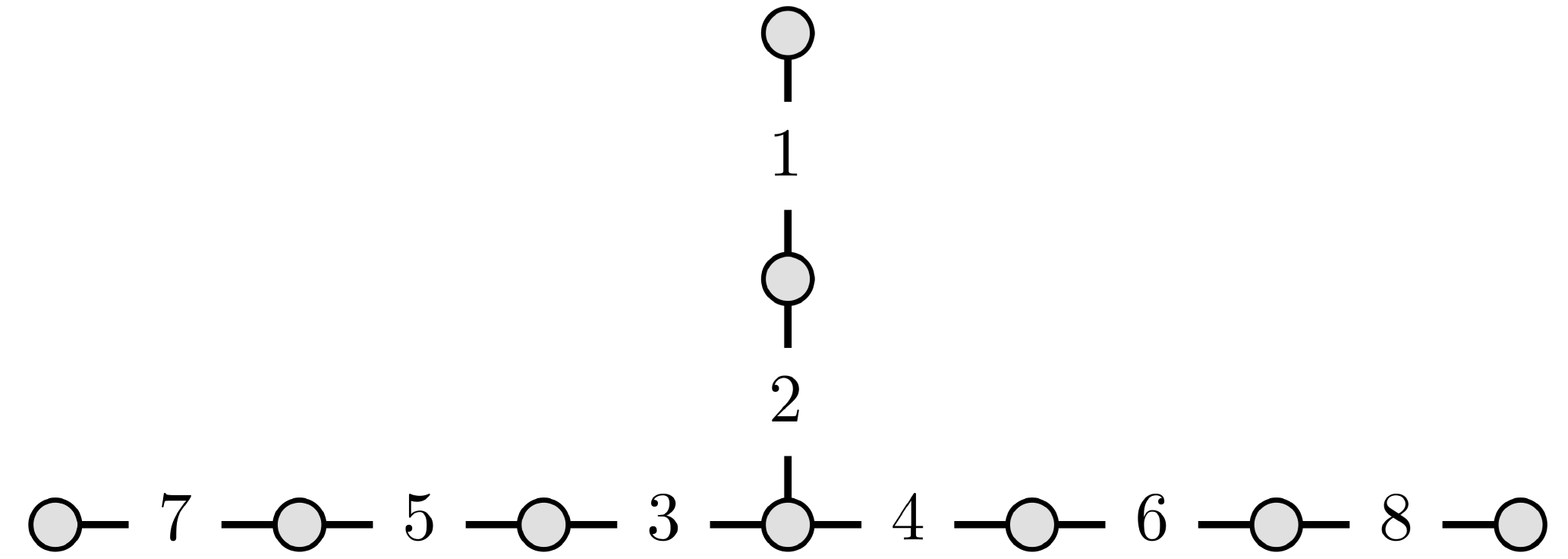}
\caption{A minimal example suggesting Lemma \ref{thm:necessity-2}.}
\label{fig:CounterExample2}
\end{figure}

\begin{example}
\label{ex:2}
Let $\Gamma=(N,\gamma)$ be the independent set game defined on the graph in Figure \ref{fig:CounterExample2}.
Clearly, every vertex in the underlying graph has distance at most two to a pendant vertex.
However, $\Gamma$ is not population monotonic.
Assume to the contrary that $\{\boldsymbol{x}_S\}_{S\in 2^{N} \backslash \{\emptyset\}}$ is a PMAS of $\Gamma$.
Notice that edges $3$ and $4$ are non-pendant edges of type I, and edge $2$ is a non-pendant edge of type II.
Lemma \ref{thm:NonPendantEdgeI} implies that $x_{234,3}=x_{234,4}=1$ and Lemma \ref{thm:NonPendantEdgeII} implies that $x_{234,2}=0$.
It follows that
\begin{equation*}
2=x_{234,2}+x_{234,3}+x_{234,4}>\gamma (234)=1,
\end{equation*}
which contradicts efficiency of PMASes.
$\lhd$
\end{example}

Based on Example \ref{ex:2}, we introduce another structural characterization for population monotonic independent set games.

\begin{lemma}
\label{thm:necessity-2}
If $\Gamma_G$ is population monotonic, then non-pendant edges of type I are not incident in $G$.
\end{lemma}
\begin{proof}
Assume to the contrary that $i_1$ and $i_2$ are non-pendant edges of type I incident to the same vertex $v$ in $G$.
By Lemma \ref{thm:necessity-1}, $v$ has distance two to pendant vertices in $G$, suggesting $\gamma \big(\delta (v)\big)=1$.
Lemma \ref{thm:NonPendantEdgeI} implies that $x_{\delta (v), i_1}=x_{\delta (v), i_2}=1$.
It follows that
\begin{equation*}
\sum_{i\in \delta (v)}x_{\delta (v), i}\geq x_{\delta (v),i_1}+x_{\delta(v),i_2}>\gamma \big(\delta (v)\big)=1,
\end{equation*}
which contradicts efficiency of PMASes.
\end{proof}

\subsection{A necessary and sufficient characterization}

Both Lemmas \ref{thm:necessity-1} and \ref{thm:necessity-2} are necessary conditions for the population monotonicity of independent set games.
We show that combining these two conditions yields a sufficient condition for the population monotonicity of independent set games.

\begin{lemma}\label{thm:sufficiency}
If $G$ satisfies the following two conditions:
\begin{enumerate}[label={\emph{($\roman*$)}}]
	\item\label{itm:necessity-1} every vertex has  distance at most two to pendant vertices;
	\item\label{itm:necessity-2} non-pendant edges of type I are not incident,
\end{enumerate}
then $\Gamma_G$ is population monotonic.
\end{lemma}

\begin{proof}
Assume that $G$ satisfies conditions \emph{$\ref{itm:necessity-1}$} and \emph{$\ref{itm:necessity-2}$}.
Let $E_0$, $E_1$ and $E_2$ denote the set of pendant edges, non-pendant edges of type I, and non-pendant edges of type II, respectively.
For any non-empty $S\subseteq N$, define $\boldsymbol{x}_S=\{x_{S,i}\}_{i\in S}$ by
\begin{equation}
\label{eq:mech}
x_{S,i} =
    \begin{cases*}
      	\hspace{3mm} 1 & if $i\in E_0$,  \\
      	\hspace{3mm} 1 & if $i\in E_1$ with $\delta (u)\subseteq S$ or $\delta (v)\subseteq S$,\\
	\frac{1}{\lvert \delta(u)\rvert} & if $i\in E_2$ with $\delta (u)\subseteq S\cap E_2$,\\
	\frac{1}{\lvert \delta(v)\rvert} & if $i\in E_2$ with $\delta (v)\subseteq S\cap E_2$,\\
	\hspace{3mm} 0 & otherwise,
    \end{cases*}
\end{equation}
where $u$ and $v$ are the endpoints of edge $i$.
Notice that $\delta (u)\subseteq S\cap E_2$ and $\delta (v)\subseteq S\cap E_2$ never occur simultaneously when $i\in E_2$.
Hence every $\boldsymbol{x}_S$ is well-defined.
We show that $\{\boldsymbol{x}_S\}_{S\in 2^N\backslash \{\emptyset\}}$ is a PMAS of $\Gamma_G$.

It is easy to verify the monotonicity condition.
Let $S\subseteq T\subseteq N$ and $i\in S$.
It is trivial that $x_{S,i}\leq x_{T,i}$ when $x_{S,i}=0$.
Hence assume that $x_{S,i}>0$.
By definition, either $i\in E_0$ or $i\in E_1$ implies $x_{S,i}=x_{T,i}=1$.
If $i\in E_2$, then $x_{S,i} = x_{T,i}=\frac{1}{\lvert  \delta (v)\rvert } $ follows from definition, where $v$ is an endpoint of edge $i$.
Hence $x_{S,i} = x_{T,i}$ when $x_{S,i}>0$.
Therefore, the monotonicity follows.

It remains to check the efficiency condition.
Let $S\subseteq N$ and $I_S$ be a maximum independent set of $G[V \langle S\rangle]$ that contains as many pendant vertices of $G$ as possible.
We claim that every non-pendant vertex from $I_S$ has distance two to pendant vertices in $G$.
Assume to the contrary that there exists a non-pendant vertex $v\in I_S$ adjacent to a pendant vertex $u$ in $G$.
Notice that $v \in I_S$ implies $u \in V\langle S\rangle \backslash I_S$.
Then $I_S -v+u$ is a maximum independent set of $G[V\langle S\rangle]$ with more pendant vertices than $I_S$, which contradicts the selection of $I_S$.
In the following, we show
\begin{equation}\label{eq:efficiency}
  \sum_{i\in S} x_{S,i}=\sum_{v\in I_S} \sum_{i\in \delta (v)} x_{S,i}=\lvert I_S\rvert =\gamma (S).
\end{equation}

We first prove $\sum_{i\in \delta (v)}x_{S,i}=1$ for any $v\in I_S$.
Let $v\in I_S$.
Clearly, $\delta (v)\subseteq S$.
It is trivial that $\sum_{i\in \delta (v)}x_{S,i}=1$ if $v$ is pendant in $G$.
Hence assume that $v$ is non-pendant in $G$.
It follows that $v$ has distance two to pendant vertices in $G$.
We show $\sum_{i\in \delta (v)}x_{S,i}=1$ by distinguishing two cases of $\delta (v)$.
If $\delta (v)\subseteq E_2$, then $x_{S,i}=\frac{1}{\lvert \delta (v)\rvert}$ for any $i\in \delta (v)$, implying that $\sum_{i\in \delta (v)}x_{S,i}=1$.
Otherwise, let $i^*\in \delta (v)$ be a non-pendant edge of type I.
By definition, $x_{S,i^*}=1$ and $x_{S,i}=0$ for any $i\in \delta(v)\backslash i^*$.
It follows that $\sum_{i\in \delta (v)}x_{S,i}=1$.

We now prove $x_{S,i}=0$ for any $i\in S$ not covered by $I_S$.
Let $i\in S$ be an edge not covered by $I_S$.
Let $u$ and $v$ be the endpoints of edge $i$.
Clearly, $u,v\not\in I_S$.
By definition, $x_{S,i}=0$ follows trivially if neither $\delta (u)\subseteq S$ nor $\delta (v) \subseteq S$.
Hence without loss of generality, assume that $\delta (v)\subseteq S$.
According to the selection of $I_S$, edge $i$ is non-pendant since every pendant edge in $S$ is covered by $I_S$.
We claim that $i\in E_2$, i.e., $i$ is a non-pendant edge of type II.
Assume to the contrary that $i\in E_1$.
Then both $u$ and $v$ have distance two to pendant vertices in $G$.
Moreover, $u$ is the unique neighbor of $v$ that has distance two to pendant vertices in $G$.
Since every non-pendant vertex from $I_S$ has distance two to pendant vertices in $G$,
$I_S+v$ is an independent set of $G[V\langle S\rangle]$, which contradicts the maximality of $I_S$.
Hence we have $i\in E_2$ as asserted.
Without loss of generality, assume that $v$ has distance two to pendant vertices in $G$.
It is trivial that $\delta (u)\not\subseteq E_2$, as $u$ is adjacent to pendant vertices in $G$.
To prove $x_{S,i}=0$, it suffices to show that $\delta (v)\not\subseteq E_2$.
Assume to the contrary that $\delta (v)\subseteq E_2$.
Then every neighbor of $v$ is adjacent to pendant vertices in $G$.
According to the selection of $I_S$, $I_S +v$ is an independent set of $G[V\langle S\rangle]$, which contradicts the maximality of $I_S$.
Hence we have $\delta (v)\not\subseteq E_2$.
By definition, $x_{S,i}=0$.

Therefore, the efficiency follows from \eqref{eq:efficiency}.
\end{proof}

\begin{figure}[h]
\centering
\includegraphics[width=.4\textwidth]{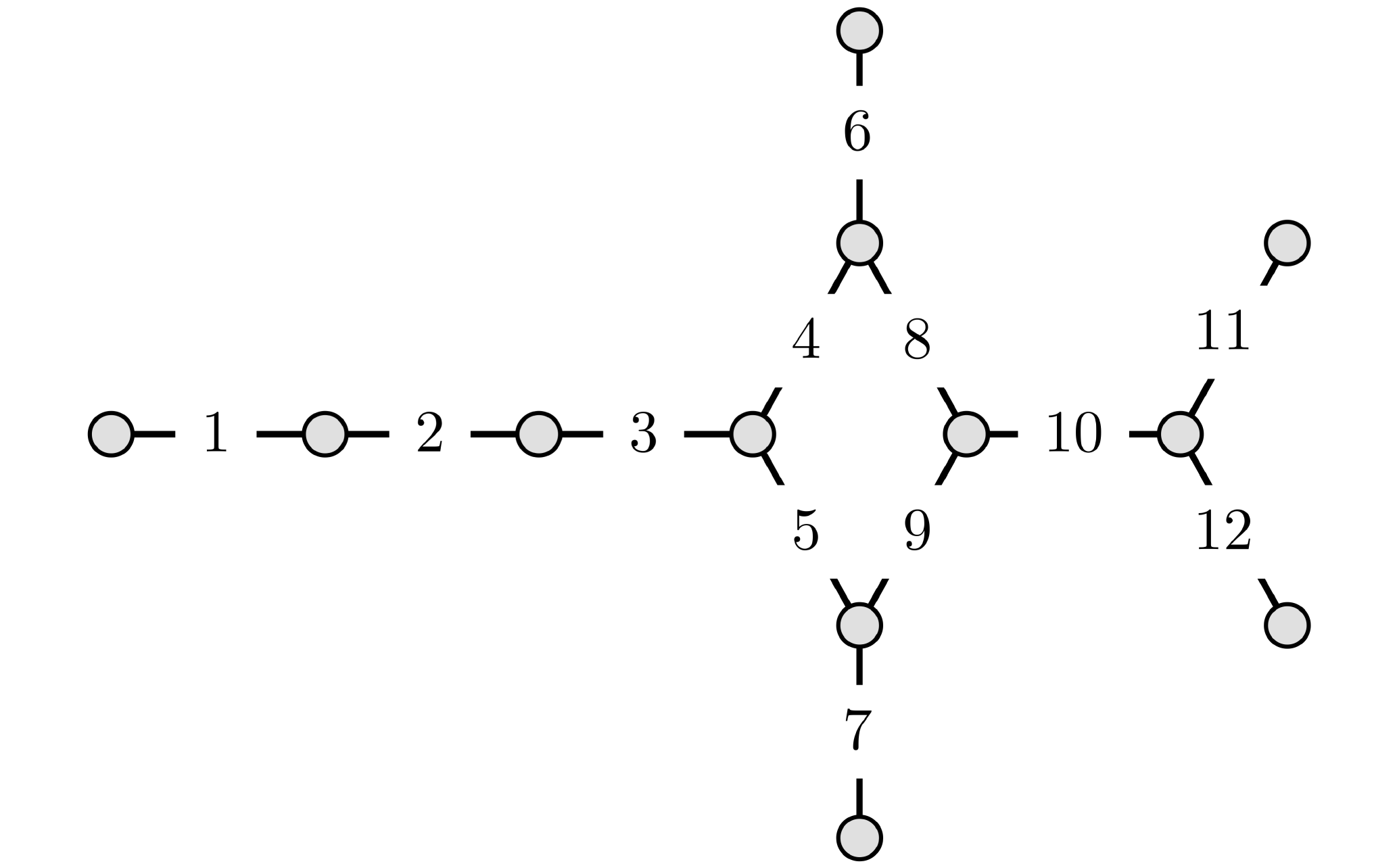}
\caption{An illustrative example.}
\label{fig:example}
\end{figure}

\begin{example}
\label{ex:3}
Let $\Gamma=(N,\gamma)$ be the independent set game defined on the graph in Figure \ref{fig:example}.
It is easy to verify that the underlying graph satisfies both conditions of Lemma \ref{thm:sufficiency}.
Let $E_0$, $E_1$, $E_2$ denote the set of pendant edges, non-pendant edges of type I and non-pendant edges of type II, respectively.
Note that $E_0=\{1,6,7,11,12\}$, $E_1=\{3\}$, $E_2=\{2,4,5,8,9,10\}$.
For any $S\subseteq N$, define $\boldsymbol{x}_S=(x_{S,i})_{i\in S}$ according to \eqref{eq:mech}, i.e.,
$x_{S,i}=1$ if $i\in S\cap \{1,6,7,11,12\}$;
$x_{S,3}=1$ if $\{2,3\}\subseteq S$ or $\{3,4,5\}\subseteq S$;
$x_{S,8}=x_{S,9}=x_{S,10}=\frac{1}{3}$ if $\{8,9,10\}\subseteq S$;
and $x_{S,i}=0$ otherwise.
We may verify that $\{\boldsymbol{x}_S\}_{S\in 2^N\backslash \{\emptyset\}}$ defined above is indeed a PMAS of $\Gamma$.
$\lhd$
\end{example}

Finally, we come to our main result.
Combining Lemmas \ref{thm:necessity-1}, \ref{thm:necessity-2} and \ref{thm:sufficiency} gives a necessary and sufficient characterization for the population monotonicity of independent set games.

\begin{theorem}
\label{thm:main}
An independent set game is population monotonic if and only if the underlying graph satisfies the following two conditions:
\begin{enumerate}[label={\emph{($\roman*$)}}]
	\item every vertex has distance at most two to pendant vertices;
	\item non-pendant edges of type I are not incident.
\end{enumerate}
\end{theorem}

Both conditions of Theorem \ref{thm:main} can be verified efficiently.
The first condition can be verified with the BFS algorithm.
Indeed, a vertex has distance at most two to pendant vertices if and only if there is a pendant vertex in the first two levels of the BFS tree rooted at the vertex.
Moreover, the class of a non-pendant edge is determined after BFS searching both endpoints of the edge.
Hence the second condition can be verified easily after checking the first condition for every vertex with the BFS algorithm.
Therefore, we have the following corollary.

\begin{corollary}
The population monotonicity of an independent set game can be determined efficiently.
\end{corollary}

\section{Discussion}
\label{sec:discussion}

It is well-known that convex games are population monotonic \cite{Spru90}, but the converse may not be true.
Xiao et al. \cite{XWF21} show that independent set games are convex if and only if every non-pendant edge is incident to a pendant edge in the underlying graph.
In fact, the condition of convexity given by Xiao et al. \cite{XWF21} is a special case of the first condition of population monotonicity in Theorem \ref{thm:main}.
Indeed, pendant edges/vertices are necessary for both convexity and population monotonicity.
And a graph meeting the condition of convexity clearly satisfies the first condition of population monotonicity.
However, population monotonicity allows non-pendant edges not to be incident to pendant edges provided that both endpoints of every such non-pendant edge have distance at most two to pendant vertices.

One possible direction for future work is to characterize the PMASes of other combinatorial optimization games such as flow games and linear production games.
Besides, related concepts of PMASes, such as link monotonic allocation schemes \cite{Sli05} and monotonic core allocation paths \cite{Al19} are also worth studying.

\section*{Acknowledgement}
We would like to thank the reviewers for their valuable comments and suggestions.
Han Xiao is supported in part by the National Natural Science Foundation of China (Nos.\,12001507, 11871442 and 12171444) and the Natural Science Foundation of Shandong (No.\,ZR2020QA024).
Donglei Du is supported in part by the Natural Sciences and Engineering Research Council of Canada (No.\,283106) and the National Natural Science Foundation of China (Nos.\,11771386 and 11728104).
Dachuan Xu is supported in part by the National Natural Science Foundation of China (No.\,12131003) and the Natural Science Foundation of Beijing (No.\,Z200002).

\bibliographystyle{abbrv}
\bibliography{reference}

\begin{thebibliography}{10}

\bibitem{Al19}
T.~Abe and S.~Liu.
\newblock Monotonic core allocation paths for assignment games.
\newblock {\em Social Choice and Welfare}, 53(4):557--573, 2019.

\bibitem{DFS06}
X.~Deng, Q.~Fang, and X.~Sun.
\newblock Finding nucleolus of flow game.
\newblock In {\em {Proc.} {SODA}}, page 124–131, 2006.

\bibitem{DIN99}
X.~Deng, T.~Ibaraki, and H.~Nagamochi.
\newblock Algorithmic aspects of the core of combinatorial optimization games.
\newblock {\em Mathematics of Operations Research}, 24(3):751--766, 1999.

\bibitem{DINZ00}
X.~Deng, T.~Ibaraki, H.~Nagamochi, and W.~Zang.
\newblock Totally balanced combinatorial optimization games.
\newblock {\em Mathematical Programming}, 87(3):441--452, 2000.

\bibitem{DLN13}
M.~Dobson, V.~Leoni, and G.~Nasini.
\newblock A characterization of edge-perfect graphs and the complexity of
  recognizing some combinatorial optimization games.
\newblock {\em Discrete Optimization}, 10(1):54--60, 2013.

\bibitem{ELN12}
M.~Escalante, V.~Leoni, and G.~Nasini.
\newblock A graph theoretical model for the total balancedness of combinatorial
  games.
\newblock {\em Rev. Un. Mat. Argentina}, 53(1):85--92, 2012.

\bibitem{FK93}
U.~Faigle and W.~Kern.
\newblock On some approximately balanced combinatorial cooperative games.
\newblock {\em ZOR - Methods and Models of Operations Research},
  38(2):141--152, 1993.

\bibitem{IMM08}
N.~Immorlica, M.~Mahdian, and V.~S. Mirrokni.
\newblock Limitations of cross-monotonic cost-sharing schemes.
\newblock {\em ACM Transactions on Algorithms}, 4(2):1--25, 2008.

\bibitem{KZ82}
E.~Kalai and E.~Zemel.
\newblock Totally balanced games and games of flow.
\newblock {\em Mathematics of Operations Research}, 7(3):476--478, 1982.

\bibitem{KP03}
W.~Kern and D.~Paulusma.
\newblock Matching games: The least core and the nucleolus.
\newblock {\em Mathematics of Operations Research}, 28(2):294--308, 2003.

\bibitem{KS20}
Z.~K. Koh and L.~Sanit\'{a}.
\newblock An efficient characterization of submodular spanning tree games.
\newblock {\em Mathematical Programming}, 183(1-2):359--377, 2020.

\bibitem{KPT20}
J.~K\"{o}nemann, K.~Pashkovich, and J.~Toth.
\newblock Computing the nucleolus of weighted cooperative matching games in
  polynomial time.
\newblock {\em Mathematical Programming}, 183(1-2):555--581, 2020.

\bibitem{LXF21}
B.~Liu, H.~Xiao, and Q.~Fang.
\newblock A combinatorial characterization for population monotonic allocations
  in convex independent set games.
\newblock {\em Asia-Pacific Journal of Operational Research}, 38(5):2140006,
  2021.

\bibitem{NMT04}
H.~Norde, S.~Moretti, and S.~Tijs.
\newblock Minimum cost spanning tree games and population monotonic allocation
  schemes.
\newblock {\em European Journal of Operational Research}, 154(1):84--97, 2004.

\bibitem{SS71}
L.~S. Shapley and M.~Shubik.
\newblock The assignment game {I}: The core.
\newblock {\em International Journal of Game Theory}, 1(1):111--130, 1971.

\bibitem{Sli05}
M.~Slikker.
\newblock Link monotonic allocation schemes.
\newblock {\em International Game Theory Review}, 7(4):473--489, 2005.

\bibitem{Spru90}
Y.~Sprumont.
\newblock Population monotonic allocation schemes for cooperative games with
  transferable utility.
\newblock {\em Games and Economic Behavior}, 2(4):378--394, 1990.

\bibitem{Vazi22}
V.~V. Vazirani.
\newblock The general graph matching game: Approximate core.
\newblock {\em Games and Economic Behavior}, 132:478--486, 2022.

\bibitem{XF22}
H.~Xiao and Q.~Fang.
\newblock Population monotonicity in matching games.
\newblock {\em Journal of Combinatorial Optimization}, 43:699--709, 2022.

\bibitem{XWF21}
H.~Xiao, Y.~Wang, and Q.~Fang.
\newblock On the convexity of independent set games.
\newblock {\em Discrete Applied Mathematics}, 291:271--276, 2021.

\end{thebibliography}
\end{document}